\documentclass[runningheads]{llncs}
\usepackage{latexsym}
\usepackage{algorithm}
\usepackage{algorithmic}
\usepackage{amsmath}
\usepackage{amssymb}

\usepackage{mathtools}
\usepackage{amsfonts}%,amsthm}
\usepackage{graphicx}
\usepackage{dsfont}
\usepackage[svgnames]{xcolor}

\def\fl{\left\lfloor \frac{d-1}{2} \right\rfloor}
\def\del{\delta}
\def\lam{\lambda}
\long\def\comment#1{}

\def\R{{\mathbb{R}}}
\def\Pmax{P_{\max}}
\def\PNB{P_{\text{NB}}}
\def\Preg{P_{\text{reg}}}

\def\NB{\mathrm{NB}}

\newcommand\red[1]{\textcolor{red}{#1}}
\newcommand\blue[1]{\textcolor{blue}{#1}}
\newcommand\green[1]{\textcolor{ForestGreen}{#1}}

\newcommand{\seclab}[1]{\label{sec:#1}}                     % normal
\newcommand{\tablab}[1]{\label{tab:#1}}                         % normal
\newcommand{\SEC}[1]{Section \ref{sec:#1}}       % SECtion ref  - Section 1.1
\newcommand{\TAB}[1]{Table~\ref{tab:#1}}        % TABle ref   - Table 1.1

\begin{document}

\title{Faster Algorithms for Next Breakpoint and
  Max Value for Parametric Global Minimum Cuts} 

\titlerunning{Faster Algorithms for Parametric Global Minimum Cut Problems}

\institute{}

\author{\textsc{
		Hassene Aissi\inst{1} \and
		S. Thomas McCormick\inst{2} \and
		Maurice Queyranne\inst{2}}}

\institute{Paris Dauphine University.
                  {\tt aissi@lamsade.dauphine.fr} \and Sauder School of Business
                  at the University of British Columbia.  {\tt
                    \{tom.mccormick,maurice.queyranne\}@sauder.ubc.ca}}

\maketitle

\begin{abstract}

  The parametric global minimum cut problem concerns a graph $G = (V,
  E)$ where the cost of each edge is an affine function of a parameter
  $\mu \in \mathbb{R}^d$ for some fixed dimension $d$. We consider the
  problems of finding the next breakpoint in a given direction, and
  finding a parameter value with maximum minimum cut value.  We develop
  strongly polynomial algorithms for these problems that are faster
  than a naive application of Megiddo's parametric search technique.
  Our results indicate that the next breakpoint problem is easier than
  the max value problem.

\end{abstract}

\section{Introduction}\label{intro}

Connectivity is a central subject in graph theory and has many practical applications in, e.g., communication and electrical networks. We consider the parametric global minimum cut problem in graphs. A {\em cut} $X$ in an undirected graph $G = (V,E)$ is a non-trivial vertex subset, i.e.,
$\emptyset \ne X\subset V$. It cuts the set $\delta(X) = \{e\in E : e \cap X\neq \emptyset \neq e\setminus X\}$ of edges.

In the parametric global minimum cut problem, we are given an undirected graph $G = (V, E)$ where the cost $c_\mu(e)$ of each edge $e \in E$ is an affine function of a
$d$-dimensional parameter $\mu \in \mathbb{R}^d$, i.e., $c_\mu(e) =
c^0(e)+ \sum_{i=1}^{d}\mu_i c^i(e)$, where $c^0,\ldots,c^d:E
\rightarrow \mathbb{Z}$ are $d+1$ cost functions defined on the set of
edges. By not imposing a sign condition on these functions, we may
handle, as in~\cite[Section 3.5]{Rad93}, situations where some
characteristics, measured by functions $c^i$, improve with $\mu$ while
other deteriorate. We assume that the dimension $d$ of the parameter
space is a fixed constant. 
The \textit{cost} of cut $C$ for the edge costs $c_{\mu}$ is
$c_{\mu}(C) \equiv c_{\mu}(\delta(C)) = \sum_{e \in \delta(C)}
c_\mu(e)$. Define $M_0 = \{ \mu\in \R^d \mid c_\mu(e)\ge 0$ for all $e\in E\}$ a closed and convex subset of the parameter space where the parametric costs of all the edges are non-negative. Throughout the paper
we consider only $\mu$ belonging to a nonempty simplex $M \subset
M_0$, as negative edge costs usually lead to NP-hard minimization
problems (see \cite{McRinaldiRao}).  As usual we denote $|E|$
by $m$ and $|V|$ by $n$.
 
For any $\mu \in M$, let $C^*_\mu$ denote a cut with a minimum cost $Z(\mu) \equiv c_\mu(C^*_\mu)$ for edge costs $c_\mu$. Function $Z := Z(\mu)$ is a piecewise linear concave function~\cite{NW99}.  Its graph is composed by a number of facets (linear pieces) and breakpoints at which $d$ facets meet.
In order to avoid dealing with a trivial problem, $Z$ is
assumed to have at least one breakpoint. The maximum number of facets (linear pieces) of the graph of $Z$ is called the \emph{combinatorial facet complexity} of $Z$. Mulmuley~\cite[Theorem~3.10]{Mul99} considers
the case $d=1$ and gives a super-polynomial bound on the combinatorial
facet complexity of the global minimum cut problem.
In~\cite[Theorem~4]{AMTQMP}, the authors extended this result to a
constant dimension $d$ and give a strongly polynomial bound
$O\Big(m^d n^2 \log^{d-1}n\Big)$. By combining this result with
several existing computational geometry algorithms, the authors give
an $O(m^{d\fl}\ \ n^{2 \fl}\ \ \log^{(d-1)\fl + O(1)}n)$ time
algorithm for constructing function $Z$ for general $d$, and a $O(mn^4 \log n +
  n^5\log^2 n)$ algorithm when $d = 1$. In the particular case where
cost functions $c^0,\ldots,c^d$ are nonnegative, Karger~\cite{Kar16}
gives a significantly tighter bound $O\Big(n^{d+2}\Big)$ on the
combinatorial facet complexity and shows that function $Z$ can be
computed using a randomized algorithm in $O\Big(n^{2d+2} \log n \Big)$ time.  These results are
summarized in rows 5 and 6 of \TAB{summary}.

In this paper, we consider the following parametric problems:

\begin{enumerate}
\item[$\PNB(M)$] Given a polyhedron $M \subset \mathbb{R}^d$, a value $\mu^0 \in M \subset \mathbb{R}^d$, and a
direction $\nu \in \mathbb{Z}^d$. Find the
next breakpoint $\mu^\NB \in M$ of $Z$ after $\mu^0$ in direction $\nu$, if any. 

\item[$\Pmax(M)$] Given a polyhedron $M \subset \mathbb{R}^d$, find a value $\mu^* \in M $ such that $Z(\mu^*) = \max_{\mu \in M} Z(\mu)$. 
\end{enumerate}

In contrast to $\Pmax(M)$, $\PNB(M)$ is a one-dimensional parametric
optimization problem as it considers the restriction of function $Z$
to some direction $\nu \in \mathbb{Z}^d$. This problem corresponds to
the \emph{ray shooting problem} which is a standard topic in
sensitivity analysis \cite[Section 30.3]{FV07} to identify ranges of
optimality and related quantities.  Given $\lambda \geqslant 0$, the
cost $c_{\mu^0+\lambda \nu}$ of each edge $e \in E$ in the direction
$\nu$ is defined by $c_{\mu^0+\lambda \nu}(e) = c^0(e)+
\sum_{i=1}^{d}(\mu^0_i+\lambda \nu_i) c^i(e)$. Let $\bar{c}^0(e) =
c^0(e) + \sum_{i=1}^{d} \mu^0 c^i(e)$ and $\bar{c}^1(e) =
\sum_{i=1}^{d} \nu_i c^i(e)$. The edge costs can be rewritten as
$c_{\mu^0+\lambda \nu}(e) = \bar{c}^0(e) + \lambda \bar{c}^1(e)$. For
any cut $\emptyset \neq C \subset V$, its cost for the edge costs
$c_{\mu^0+\lambda \nu}$ is a function $c_{\mu^0+\lambda \nu}(C) =
\bar{c}^0(C) + \lambda \bar{c}^1(C)$ of variable $\lambda$.  For any
$\lam\ge 0$ let $Z'(\lam) = Z'(\mu^0+\lam\nu, \nu)$ denote the right
derivative of $Z(\mu^0 + \lam \nu)$ in direction $\nu$ at $\lam$.  If
the next breakpoint $\mu^\NB$ exists, define $\lam^\NB$ by $\mu^\NB =
\mu^0 + \lam^\NB \nu$, and let $C^*_{\mu^\NB}$ denote an optimal cut
for edge costs $c_{\mu^\NB}(e)$ defining the slope $Z'(\lam^\NB)$.

$\Pmax(M)$ arises in the context of network reinforcement problem.  Consider the following 2-player game of
reinforcing a graph against an attacker. Given a graph $G=(V,E)$ where each edge $e \in E$ has a capacity $c^0(e)$, the Graph player wants to
reinforce the capacities of the edges in $E$ by buying $d+1$ resources subject to a budget $B$.
The Graph player can spend $\$\mu_i \geqslant 0$ on each resource $i$ to increase the capacities of all edges to $c_\mu(e) = c^0(e) + \sum_{i = 1}^{d+1} \mu_i c^i(e)$, where all functions $c^i$ are assumed to be non-negative. The Attacker wants to remove some edges of $E$ in order to cut the graph into two  pieces at a minimum cost. Therefore, these edges correspond to an optimal cut $\delta(C^*_\mu)$ and their removal cost is $Z(\mu)$.  The Graph player wants to make it as expensive as
possible for the Attacker to cut the graph, and so he wants to
solve $\Pmax$. It is
optimal for the Graph player to spend all the budget, and thus to spend $\mu_{d+1} = B - \sum_{i = 1}^d \mu_i$ on resource $d+1$. Therefore, the cost of removing
edge $e$ as a function of the amounts spent on the first $d$ resources
is $c_\mu(e) = c^0(e) + \sum_{i = 1}^d \mu_i c^i(e) + \big(B -
\sum_{i=1}^d \mu_i\big) c^{d+1}(e) = \big(c^0(e) + B c^{d+1}(e)\big) +
\sum_{i=1}^d \mu_i \big(c^i(e) - c^{d+1}(e)\big)$.  Note that $c^i(e) - c^{d+1}(e)$ may be negative. 
This application illustrates how negative parametric edge costs may arise even when all original data are non-negative.  

Clearly problems
$\Pmax(M)$ and $\PNB(M)$ can be solved by constructing function
$Z$. However, the goal of this paper is to give much faster
strongly polynomial algorithms for these problems without explicitly
constructing the whole function $Z$.

\subsection{Related works}

The results mentioned in this section are summarized in the first four
rows of \TAB{summary}.  We concentrate on {\em strongly} polynomial
bounds here as is the case in most of the literature, but see
\SEC{oracle} for one case where there is a potentially faster weakly
polynomial bound.

\begin{table}[htbp]
  \centering
\begin{tabular}{|r||l|l|} \hline
Problem                                          & Deterministic & Randomized \\ \hline
Non-param Global MC                            & \cite{NI92,SW97} \green{$O(mn + n^2 \log n)$}  &
\cite{Kar00} \green{${\tilde O}(m)$} (\cite{KS96} \green{${\tilde O}(n^2)$}) \\
All $\alpha$-approx for $\alpha < \frac{4}{3}$ & \cite{NI08} $O(n^4)$ & \cite{KS96} ${\tilde O}(n^2)$ \\ \hline
Megiddo $\PNB$ (${}\sim d = 1$)                            & \red{\cite{SW97}} \blue{$O(n^5 \log n)$} & \cite{Tok01,KS96} \blue{$O(n^2 \log^5 n)$} \\
Megiddo $\Pmax$ (${}\sim{}$ gen'l $d$)                          & \red{\cite{SW97}} \blue{$O(n^{2d+3} \log^d n)$}
& \cite{Tok01,KS96} \blue{$O(n^2 \log^{4d+1} n)$} \\ \hline
% \# facets $Z(\mu)$                        & $O(n^{d+2})$ \cite{Kar00} &
% $O(n^{d+2})$ \cite{Kar00} \\ 
All of $Z(\mu)$ for $d = 1$                       & \cite{AMTQMP} \blue{$O(mn^4 \log n +
  n^5\log^2 n)$}  & \cite{Kar16} \blue{$O(n^{4} \log
n)$} \\
All of $Z(\mu)$ for gen'l $d$                       & \cite{AMTQMP} \blue{[big]} & \cite{Kar16} \blue{$O(n^{2d+2} \log
n)$} \\ \hline \hline
This paper $\PNB$ (${}\sim d = 1$)         & \cite{NI92,SW97} \red{$O(mn + n^2\log
n)$} & \cite{Kar93} \red{\bf $O(n^2 \log^3 n)$} \\ 
This paper $P_{\max}$ (${}\sim{}$gen'l $d$)   & \red{$O(n^4 \log^{d-1}
n)$} & \red{???} \\ \hline
\end{tabular}

\caption{New results in this paper are in
    \red{red}.  Compare these to the non-parametric lower bounds in \green{green},
    and the various upper bounds in \blue{blue}. \tablab{summary}}
\end{table}

The standard (non-parametric) global minimum cut minimum cut is a special case of the parametric global minimum cut, i.e., for some fixed value
$\mu \in M$. Nagamochi and Ibaraki~\cite{NI92} and Stoer and Wagner~\cite{SW97}
give a deterministic algorithm for this problem that 
runs in $O(m n + n^2 \log n)$ time. Karger and Stein~\cite{KS96} give a faster randomized algorithm that runs in $\tilde{O}(n^2)$ time. Karger~\cite{Kar00} improves the running time and give an $\tilde{O}(m)$ time algorithm.  

Given $\alpha > 1$, a cut is called {\em
$\alpha$-approximate} if its cost is at most at factor of $\alpha$ larger than the optimal value. A remarkable property of the global minimum cut problem is that there exists a strongly polynomial number of near-optimal cuts.  Karger~\cite{Kar00}
showed that the number of $\alpha$-approximate cuts is
$O(n^{\lfloor2 \alpha\rfloor})$.  Nagamochi et al.~\cite{NNI97} give
a deterministic $O(m^2n + mn^{2 \alpha})$ time algorithm for
enumerating them. For the particular case $1 < \alpha < \frac{4}{3}$, they
improved this running time to $O(m^2n + mn^2 \log
n)$. Nagamochi and Ibaraki~\cite[Corollary 4.14]{NI08} further
reduced the running time to $O(n^4)$. The fastest randomized algorithm to enumerate all the near-optimal cuts, which is an $\tilde{O}(n^{\lfloor2 \alpha\rfloor})$ time algorithm by Karger and Stein\cite{KS96}, is faster than the best deterministic algorithm.

Megiddo's parametric searching method~\cite{Meg79,Meg83} is a powerful
technique to solve parametric optimization problems. Megiddo's
approach was originally designed to handle one-dimensional parametric
problems. Cohen and Megiddo~\cite{CM93} extend it to fixed dimension
$d > 1$, see also \cite{AS98}. The crucial requirement is that the
underlying non-parametric problem must have an \textit{affine}
algorithm, that is all numbers manipulated are affine functions of
parameter $\mu$. This condition is not restrictive, as many
combinatorial optimization algorithms have this property; e.g.,
minimum spanning tree~\cite{Fer00}, matroid and polymatroid
optimization~\cite{Tok01}, maximum flow~\cite{CM94}. The technique can
be summarized as follows in the special case $d=1$. Megiddo's approach
simulates the execution of an affine algorithm $\mathcal{A}$ on an
unknown target value $\bar{\mu}$ ($= \mu^\NB$ or $\mu^*$) by
considering it as a symbolic constant. During the course of execution
of $\mathcal{A}$, if we need to determine the sign of some function
$f$ at $\bar{\mu}$, we compute the root $r$ of $f$. The key point is
that by testing if $\bar{\mu} = r$, $\bar{\mu} < r$, or $\bar{\mu} >
r$, we can determine the sign of $f(\bar{\mu})$. This operation is
called a \textit{parametric test} and requires calling algorithm
$\mathcal{A}$ with parameter value fixed at $r$.

Tokuyama~\cite{Tok01} considers the analogue of problem $\Pmax(M)$ for
several geometric and graph problems, called the \textit{minimax (or
  maximin) parametric optimization}, and give efficient algorithms for
them based on Megiddo's approach. He observes that the randomized
algorithm of Karger~\cite{Kar93} is affine. In order to improve the
running time, Tokuyama implemented Megiddo's technique using the
parallel algorithm of Karger and Stein~\cite{KS96} which solves the
minimum cut problem in $\tilde{O}(\log^3n)$ randomized parallel time
using $O(\frac{n^2}{\log^2n})$ processors. The resulting randomized
algorithm for $\Pmax(M)$ has a $O(n^2 \log^{4d+1} n)$ running
time. The result was stated only for $\Pmax(M)$ but it is easy to see
that the same running time can be obtained for $\PNB(M)$. In
Appendix~\ref{sec:Pmax_geometric}, we show that Stoer and Wagner's
algorithm~\cite{SW97} is affine and can be combined with Megiddo's
approach in order to solve $\PNB(M)$ and $\Pmax(M)$. This gives
deterministic algorithms that run in $O(n^{2d+3}\log^d n)$ and
$O(n^{5}\log n)$ time for $\Pmax(M)$ and $\PNB(M)$ respectively
(Proposition~\ref{prop:meggido}).

\subsection{Our results}  \seclab{results}

Our new results are summarized in rows 7 and 8 of \TAB{summary}.

The algorithms based on Megiddo's approach typically introduce a
slowdown with respect to the non-parametric algorithm. For $d=1$,
these algorithms perform similar parametric tests and solve problems
$\PNB(M)$ and $\Pmax(M)$ with the same running time. This gives the
impression that these problems have the same complexity in this
special case. The main contribution of the paper is to extend the
techniques of Nagamochi and Ibaraki~\cite{NI92} and Stoer and
Wagner~\cite{SW97} and Karger~\cite{Kar93} to handle parametric edge
costs.  We give faster deterministic and randomized algorithms for
problems $\PNB(M)$ and $\Pmax(M)$ which are not based on Megiddo's
approach. We show that problem $\PNB(M)$ can be solved with the same
running time as the non-parametric global minimum cut
(Theorems~\ref{thm:NB_det} and~\ref{thm:NB_rand}). We give for problem
$\Pmax(M)$ a much faster deterministic algorithm exploiting the key
property that all near-optimal cuts can be enumerated in strongly
polynomial time (Theorem~\ref{thm:max}). The algorithm builds upon a
scaling technique given in~\cite{AMTQMP}. The differences in how we
tackle problems $\PNB(M)$ and $\Pmax(M)$ illustrate that $\PNB(M)$
might be significantly easier than $\Pmax(M)$.

Notice that our new algorithms for $\PNB$ in row 7 of \TAB{summary}
are {\em optimal}, in the sense that their running times match the
best-known running times of the non-parametric versions of the problem
(up to $\log$ factors).  That is, the times quoted in row 7 of
\TAB{summary} are (nearly) the same as those in row 1 (with the
exception that we do not match the Karger's speedup from ${\tilde
  O}(n^2)$ to ${\tilde O}(m)$ for the non-parametric randomized case).

\subsection{Relating $\PNB$ and $\Pmax$}  \seclab{oracle}

Recall that $\PNB$ wants us to compute $\lam^\NB$ as this picture:

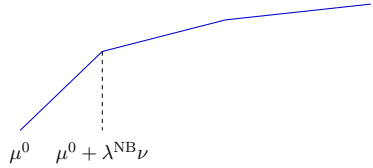
\begin{center}
\scalebox{.70}{\begin{picture}(0,0)%
\includegraphics{GMCRed1.pdftex}%
\end{picture}%
\setlength{\unitlength}{2486sp}%
\begingroup\makeatletter\ifx\SetFigFont\undefined%
\gdef\SetFigFont#1#2#3#4#5{%
  \reset@font\fontsize{#1}{#2pt}%
  \fontfamily{#3}\fontseries{#4}\fontshape{#5}%
  \selectfont}%
\fi\endgroup%
\begin{picture}(5039,2299)(-2946,-1663)
\put(-2924,-1591){\makebox(0,0)[b]{\smash{{\SetFigFont{11}{13.2}{\familydefault}{\mddefault}{\updefault}{\color[rgb]{0,0,0}$\mu^0$}%
}}}}
\put(-1754,-1591){\makebox(0,0)[b]{\smash{{\SetFigFont{11}{13.2}{\familydefault}{\mddefault}{\updefault}{\color[rgb]{0,0,0}$\mu^0+\lam^\NB\nu$}%
}}}}
\end{picture}%
}
\end{center}

Intuitively, if we ``rotate'' until the local slope at $\mu^0$ is just short of
  horizontal, then finding $\lam^\NB$ becomes equivalent to computing
  $\mu^*$ in this 1-dimensional problem:

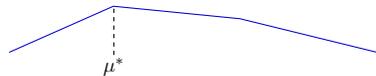
\begin{center}
\scalebox{.70}{\begin{picture}(0,0)%
\includegraphics{GMCRed2.pdftex}%
\end{picture}%
\setlength{\unitlength}{2486sp}%
\begingroup\makeatletter\ifx\SetFigFont\undefined%
\gdef\SetFigFont#1#2#3#4#5{%
  \reset@font\fontsize{#1}{#2pt}%
  \fontfamily{#3}\fontseries{#4}\fontshape{#5}%
  \selectfont}%
\fi\endgroup%
\begin{picture}(5353,1029)(-3223,-1123)
\put(-1709,-1051){\makebox(0,0)[b]{\smash{{\SetFigFont{11}{13.2}{\familydefault}{\mddefault}{\updefault}{\color[rgb]{0,0,0}$\mu^*$}%
}}}}
\end{picture}%
}
\end{center}

Unfortunately, there appears not to be any way to actually ``rotate''
the slopes of $\mu^0 + \lam \nu$ that would formalize this intuition.
We can instead consider an {\em oracle} model for $\PNB$, and for
$\Pmax$ for $d = 1$, where the algorithms interact with the graph only
through calls to an oracle with input $\mu$ that reports the local
slope at $\mu$.  We can then ask how many calls to the $\Pmax$ oracle
are necessary in order to solve $\PNB$.

In order to solve $\PNB$ using $\Pmax$'s oracle, we could proceed as
follows.  Compute the right slope $m^0 = Z'(0)$ of $\mu^0 + \lam \nu$
at $\lam = 0$ by one call to $\PNB$'s oracle. Note that $m^0$ is an
integer. Define $\del = m^0 - \frac{1}{2}$.  This means that $m^0 -
\del = \frac{1}{2}$, and (since $Z'(\lam^\NB)\le Z'(0) - 1$) that
$Z'(\lam^\NB) - \del$ at the (as-yet) unknown $\lam^\NB$ is negative.
This implies that $\lam^\NB$ solves $\Pmax$ w.r.t.\ the slopes
adjusted by subtracting $\del$.  Hence we can solve $\PNB$ by using
$\Pmax$'s oracle algorithm with its slopes adjusted downwards by
$\del$. Thus in this oracle sense $\PNB$ cannot be any harder than
$\Pmax$ for $d=1$, though it could be easier.

The Discrete Newton algorithm is one type of oracle algorithm for such
problems.  In particular, Radzik \cite[Theorem 3.9]{Rad93} shows how
to use Discrete Newton to solve $\Pmax$ for $d = 1$ (and so also
$\PNB$) in $O(m^2 \log n)$ oracle calls.  Radzik also shows a weakly
polynomial bound.  Let $C^0 = \max_e c^0(e)$, $C^1 = \max_e c^1(e)$,
and $C = \max(C^0, C^1)$.  Then \cite[Theorem 3.4, Section 3.3]{Rad93}
says that Discrete Newton solves $\Pmax$ for $d = 1$ (and so also
$\PNB$) in $O\big(\frac{\log(n C^0 C^1)}{1 + \log\log(n C^0 C^1) -
  \log\log(n C^1)}\big) \le O\big(\log(nC)\big)$ oracle calls.

We could use Discrete Newton in place of Megiddo to solve $\PNB$ by
using the method described above.  Each iteration costs $O(mn + n^2
\log n)$ time. This would give an $O(m^3n \log n + m^2n^2 \log^2 n)$
algorithm. This is slower than the $O(n^4)$ we will get from Megiddo.
However, if $\log C$ is smaller than $O(n)$, then the weakly
polynomial bound of $O\big(\log(n C)(mn + n^2 \log n)\big)$ is faster
than the $O(n^4)$ that we will get from Megiddo.

\section{Problem $\PNB(M)$}\label{sec:NB}

We discuss in Sections~\ref{sec:deterministicNB}
and~\ref{sec:deterministicNB} efficient deterministic and randomized
algorithms for solving problem $\PNB(M)$ respectively. These
algorithms are based on edge contractions. 
Before giving the algorithms, a preliminary step is to compute an upper bound $\bar{\lambda} > 0$ such that the next breakpoint $\mu^\NB$ satisfies $\mu^\NB = \mu^0  + \lambda^\NB \nu$ for some $\lambda^\NB \in [0,\bar{\lambda}]$. 
If the next breakpoint exists, then $\bar{c}^1(C^*_{\mu^0}) \neq \bar{c}^1(C^*_{\mu^\NB})$ or equivalently $|\bar{c}(C^*_{\mu^0})- \bar{c}(C^*_{\mu^\NB})| \geqslant
1$ as costs $c^i$ and the direction vector $\nu$ are in $\mathbb{Z}^d$. In this case, functions $c_{\mu^0+\lambda \nu}(C^*_{\mu^0})$ and $c_{\mu^0+\lambda \nu}(C^*_{\mu^\NB})$ intersect at 

$$
0 \leqslant \lambda^\NB = \frac{\bar{c}^0(C^*_{\mu^\NB})- \bar{c}^0(C^*_{\mu^0})}{\bar{c}^1(C^*_{\mu^0})- \bar{c}^1(C^*_{\mu^\NB})} = \frac{|\bar{c}^0(C^*_{\mu^\NB})- \bar{c}^0(C^*_{\mu^0})|}{|\bar{c}^1(C^*_{\mu^0})- \bar{c}^1(C^*_{\mu^\NB})|}
\leqslant |\bar{c}^0(C^*_{\mu^\NB})- \bar{c}^0(C^*_{\mu^0})| \leqslant \sum_{e \in E} |\bar{c}^0(e)|.
$$

\noindent Therefore, the desired upper bound is $\bar{\lambda} := \sum_{e \in E} |\bar{c}^0(e)|$.

\noindent Our algorithms require also computing the slope $Z'(\mu^0,\nu) \in \mathbb{R}$ such that for some (unknown) $\delta > 0$, $Z(\mu^0  + \lambda \nu) = Z(\mu^0) + Z'(\mu^0,\nu)\lambda$ for all $\lambda \in [0,\delta]$. 
A deterministic algorithm to compute this quantity is to call the Stoer and Wagner's
algorithm~\cite{SW97} in order to compute $Z(\mu^0+\epsilon \nu)$ for some very small $\epsilon >0$. Since this algorithm is affine (see Appendix~\ref{sec:Pmax_geometric}), one can consider $\epsilon$ as an implicit parameter with a very small value. By calling instead the affine
algorithm of Karger and Stein~\cite{KS96} in order to solve $Z(\mu^0+\epsilon \nu)$, one can obtain a randomized algorithm for computing $Z'(\mu^0,\nu)$. 

\subsection{A deterministic contraction algorithm}\label{sec:deterministicNB}

We describe in this section a deterministic algorithm for $\PNB(M)$
based on the concept of \textit{pendant pair}. We call an ordered pair
$(u,v)$ of vertices in $G$ a pendant pair for edge costs $c_\mu(e)$
for some $\mu \in M$ if $\min \{c_\mu(X): \emptyset \subset X \subset V \text{ separating } u \textit{ and } v\} = c_\mu(\delta(v))$. 

The algorithm proceeds in $n-1$ phases and computes iteratively the next breakpoint $\mu^\NB$, if any, or claims that it does not exist. In the former case, the algorithm refines, at each iteration $r$, an upper bound $\bar{\lambda}^\NB$ of $\lambda^\NB$ by choosing some $\lambda^r \in [0,\bar{\lambda}]$ and merging a pendant pair $(u^r,v^r)$ in $G^r$ for edge costs $c_{\mu^0 + \lambda^r \nu}(e)$. The process continues until the residual graph contains only one node. All the details are summarized in Algorithm~\ref{Alg_nb_detereministic}. 

\begin{algorithm}[h]
	\caption{Deterministic Parametric Edge Contraction}\label{Alg_nb_detereministic}
	\begin{algorithmic}[1]
		\REQUIRE a graph $G=(V,E)$, costs $c^0,\ldots,c^d$, a direction $\nu$, an upper bound $\bar{\lambda}$, the optimal value $Z(\mu^0)$ and a slope $Z'(\mu^0,\nu)$
		\ENSURE next breakpoint $\mu^\NB$ if any
		\STATE let $E^0 \leftarrow E$, $V^0  \leftarrow V$, $G^0\leftarrow G$, $r\leftarrow 0$, $\bar{\lambda}^\NB \leftarrow \bar{\lambda}$  
		\WHILE{$|V_r| > 1$}
		\STATE define functions $L(\lambda) := Z(\mu^0) + \lambda Z'(\mu^0,\nu)$, $Z^r(\mu) := \min_{v \in V^r} c_\mu(\delta(v))$, compute, if any,  $\hat{\lambda}^r := \min\{\lambda >0: Z^r(\mu^0 + \lambda \nu) \leqslant L(\lambda)\}$, and let $\lambda^r :=  \left\{
		\begin{array}{ll}
		\min\{\bar{\lambda},\hat{\lambda}^r\} \text{ if } \hat{\lambda}^r  \text{ exists}\\
		\bar{\lambda}  \text{ otherwise} 
		\end{array}
		\right.$ 
		\IF{$\lambda^r < \bar{\lambda}^\NB$}
			\STATE set $\bar{\lambda}^\NB \leftarrow \lambda^r$
		\ENDIF
		\STATE compute a pendant pair $(u^r,v^r)$ in $G^r$ for edge costs $c_{\mu^0 + \lambda^r \nu}(e)$ using the algorithm given in~\cite{SW97}
		\STATE merge nodes $u^r$ and $v^r$ and remove self-loops 
		\STATE set $r \leftarrow r+1$ and let $G_r=(V_r,E_r)$ denote the resulting graph
		\ENDWHILE
		\IF{$L(\bar{\lambda}^\NB) > \min_C \{c_{\mu^0  + \bar{\lambda}^\NB  \nu}(C): \emptyset\not= C \subset V_r\}$}\label{alg:NB_deterministic_test}
			\STATE return $\mu^\NB = \mu^0  + \bar{\lambda}^\NB  \nu$
		\ELSE
			\STATE the next breakpoint does not exist
		\ENDIF	
	\end{algorithmic}
\end{algorithm}

 Since cuts $\delta(v)$ for all $v \in V^r$ are also cuts in $G$, it follows that $Z(\mu) \leqslant Z^r(\mu)$ for any $\mu \in M$. In particular, $L(\lambda)=Z(\mu^0 + \lambda \nu)  \leqslant Z^r(\mu^0 + \lambda \nu)$ for any $\lambda \in [0,\lambda^\NB]$. By the definition of $\lambda^r$, this implies that 
 
 \begin{equation}\label{eq:NB_deterministic_UB}
 	\lambda^\NB \leqslant \lambda^r \text{ and } L(\lambda^r) \leqslant Z^r(\mu^0 + \lambda^r \nu).
 \end{equation}

 \begin{lemma}\label{lem:Nb_deterministic_merge}
	If the next breakpoint $\mu^\NB$ exists, then in any iteration $r$ Algorithm~\ref{Alg_nb_detereministic} either
	\begin{enumerate}
		\item[$i)$]  finds $\lambda^r = \lambda^\NB$ and returns at the end the next breakpoint, or
		\item[$ii)$] merges nodes $u^r$ and $v^r$ that are not separated by any optimal cut $C^*_{\mu^\NB}$ in $G$ for edge costs $c_{\mu^\NB}(e)$.
	\end{enumerate}

\end{lemma}
\begin{proof}
	\begin{enumerate}
		\item[$i)$] In this case, $\bar{\lambda}^\NB$ is set to  $\lambda^\NB$ at iteration $r$ and by (\ref{eq:NB_deterministic_UB}) the value of $\bar{\lambda}^\NB$ will not decrease in the subsequent iterations. Therefore, the next breakpoint is returned at the end of Algorithm~\ref{Alg_nb_detereministic}.   
		\item[$ii)$] For any iteration $r$ and any $\mu \in M$, define function $Z_{u^rv^r}(\mu) := \min_C \{c_\mu(C): \emptyset\not= C \subset V_r \text{ and separates } u^r \text{ and }  v^r \}$. By the choice of the pair $(u^r,v^r)$ and (\ref{eq:NB_deterministic_UB}), we have 
		
		$$Z_{u^rv^r}(\mu^0 + \lambda^r \nu) \geqslant Z^r(\mu^0 + \lambda^r \nu) \geqslant  L(\lambda^r) > c_{\mu^0 + \lambda^r \nu} (C^*_{\mu^\NB}),$$
		\noindent where the last inequality follows since the next breakpoint exists and function $Z$ is concave. This shows the claimed result.
	\end{enumerate}\qed
\end{proof}

\begin{lemma}\label{lem:Nb_deterministic}
	Algorithm~\ref{Alg_nb_detereministic} is correct. 
\end{lemma}
\begin{proof}
Suppose first that the next breakpoint exists. By Lemma~\ref{lem:Nb_deterministic_merge} $i)$, Algorithm~\ref{Alg_nb_detereministic} is clearly correct if $\lambda^r = \lambda^\NB$ for some $r$. Otherwise, fix an optimal cut $C^*_{\mu^\NB}$ for edge costs $c_{\mu^\NB}(e)$ defining the slope $Z'(\mu^\NB,\nu)$ at $\mu^\NB$. By Lemma~\ref{lem:Nb_deterministic_merge} $ii)$, all the pairs merged during the course of Algorithm~\ref{Alg_nb_detereministic} are not separated by $C^*_{\mu^\NB}$, and thus $C^*_{\mu^\NB}$ is associated to a cut in the final graph which is formed by a single node. This leads to a contradiction. 

Suppose now that the next breakpoint does not exist. In this case, $L(\bar{\lambda}^\NB) = \min_C \{c_{\mu^0  + \bar{\lambda}^\NB  \nu}(C): \emptyset\not= C \subset V_r\}$ and thus, Algorithm~\ref{Alg_nb_detereministic} gives a correct answer. This shows the claimed result. \qed 	
\end{proof}

Computing the lower lower envelope of $O(n)$ linear functions and getting function $Z^r$ takes $O(n\log n)$ time \cite{BM98}. Therefore, the running time of an iteration $r$ of Algorithm~\ref{Alg_nb_detereministic} is dominated by the time of computing a pendant pair in $O(m+n\log n)$ time \cite{SW97}. The added running time of the $n-1$ iterations of the while loop takes $O(mn+n^2\log n)$. Note that this corresponds to the same running time of computing a non-parametric minimum cut \cite{SW97}. Since the test performed in Step~\ref{alg:NB_deterministic_test} requires the computation of a minimum cut, it follows that the overall running time of Algorithm~\ref{Alg_nb_detereministic} is $O(mn+n^2\log n)$. The following result summarizes the running time of our contraction algorithm.

\begin{theorem}\label{thm:NB_det}
	Problem $\PNB$ can be solved in $O(mn+n^2\log n)$ time.
\end{theorem}

\subsection{A randomized contraction algorithm}\label{sec:randomNB}

The algorithm performs a number of random edge contractions and
iteratively solves the next breakpoint problem. At each iteration $r$,
the algorithm chooses some $\tilde{\mu}^r \in M$ and randomly selects
an edge $e \in E^r$ with probability
$\frac{c_{\tilde{\mu}^r}(e)}{c_{\tilde{\mu}^r}(E_r)}$ to be
contracted. The point $\tilde{\mu}^r$ is defined as the intersection
of functions $L(\lambda) := Z(\mu^0) + \lambda Z'(\mu^0,\nu)$ and
$UB^r(\lambda) := \frac{1}{|V_r|}c_{\mu^0+\lambda \nu}(E_r)$ and
may vary from one iteration to the next. The choice of the appropriate
value of $\tilde{\mu}^r$ is crucial to ensure the high success
probability of solving the problem, and is the main contribution of this algorithm.
The random edges contraction sequence continues until obtaining a graph $G'$ with two nodes.
If the next breakpoint $\mu^\NB$ exists, then the algorithm returns it after computing an optimal cut $C^*_{\mu^\NB}$ for edge costs $c_{\mu^\NB}(e)$ defining the slope $Z'(\mu^\NB,\nu)$ at $\mu^\NB$. Otherwise, the algorithm claims that it does not exist. All the details are summarized in Algorithm~\ref{Alg_nb_contract}. 

\begin{algorithm}[h]
	\caption{Randomized Parametric Random Edge Contraction}\label{Alg_nb_contract}
	\begin{algorithmic}[1]
		\REQUIRE a graph $G=(V,E)$, costs $c^0,\ldots,c^d$, a direction $\nu$, an upper bound $\bar{\lambda}$, the optimal value $Z(\mu^0)$ and a slope $Z'(\mu^0,\nu)$
		\ENSURE next breakpoint $\mu^\NB$ if any
		\STATE let $E^0 \leftarrow E$, $V^0  \leftarrow V$, $G^0\leftarrow G$, $r\leftarrow 0$
		\WHILE{$|V_r| > \rho$}\label{Alg_breakpointWhile}
		\STATE compute the intersection point $\lambda^r$  of functions $L(\lambda) := Z(\mu^0) + \lambda Z'(\mu^0,\nu)$ and $UB^r(\lambda) := \frac{1}{|V^r|}\sum_{v \in V^r} c_{\mu^0+\lambda \nu}(\delta(\{v\}))$ 
		\IF{$\lambda^r \in [0,\bar{\lambda}]$}
		\STATE \label{Alg_nb:def1} set $\tilde{\mu}^r = \mu^0 + \lambda^r \nu$
		\ELSE 
		\STATE \label{Alg_nb:def2} set $\tilde{\mu}^r = \mu^0 + \bar{\lambda} \nu$ \label{Alg_breakpointSelection2}
		\ENDIF
		\STATE  \label{Alg_nb_select} choose an arbitrary edge $e \in E_r$ with probability $\frac{c_{\tilde{\mu}^r}(e)}{c_{\tilde{\mu}^r}(E_r)}$ \label{Alg_breakpointSelection1}
		\STATE $r \leftarrow r+1$
		\STATE \label{contract} contract $e$ by merging all its vertices and removing self-loops 
		\STATE let $G_r=(V_r,E_r)$ denote the resulting graph
		\ENDWHILE
		\STATE choose uniformly at random a cut $C$ in the final graph  $G'$ and define $\mu^0 + \bar{\lambda}^\NB \nu$ as the intersection value of functions $L(\lambda)$ and $c_{\mu^0 + \lambda \nu}(C)$  
		\IF{$\bar{\lambda}^\NB >0$}\label{alg:NB_random_test}
		\STATE return $\mu^\NB = \mu^0  + \bar{\lambda}^\NB  \nu$
		\ELSE
		\STATE the next breakpoint does not exist
		\ENDIF	
	\end{algorithmic}
\end{algorithm}

We say that an edge $e$ in $G_r$ \emph{survives} at the current contraction if it is not chosen to be contracted. 
An edge $e \in G$ survives at the end of iteration $r$ if it survives
all the $r$ edge contractions. A cut $C$ \textit{survives} at the end
of iteration $r$ if every edge $e\in\delta(C)$ has survived.  We show
that a fixed optimal cut $C^*_{\mu^\NB}$ is returned by Algorithm~\ref{Alg_nb_contract} with probability at least $\dbinom{n}{2}^{-1}$.

Assume first that the next breakpoint $\mu^\NB = \mu^0 + \lambda^\NB \nu$ exists and $\lambda^\NB \in [0,\bar{\lambda}]$. Fix an optimal cut $C^*_{\mu^\NB}$ for edge costs $c_{\mu^\NB}(e)$ defining the slope $Z'(\mu^\NB,\nu)$ at $\mu^\NB$  
and suppose that it has survived until iteration $r$. Since cuts in the minor graph $G^r$ are
also cuts in $G$, it follows that 

\begin{equation}\label{eq:NB_bound1}
Z(\mu^0+\lambda \nu) \leqslant \min_C \{ c_{\mu^0+\lambda \nu}(C): \emptyset\not= C \subset V_r\} \text{ for all } \lambda \in \left[0,\lambda^\NB\right]. 
\end{equation}

\noindent Since the minimum cut value in graph $G^r$ is less than the values of all the cuts formed by singleton nodes $v \in V^r$, it follows that 

\begin{equation}\label{eq:NB_bound2}
	\min_C \{ c_{\mu^0+\lambda \nu}(C): \emptyset\not= C \subset V_r\} \leqslant \frac{1}{|V^r|}\sum_{v \in V^r} c_{\mu^0+\lambda \nu}(\delta(\{v\})) = UB(\lambda) \text{ for any } \lambda \in [0,\bar{\lambda}].
\end{equation}

\noindent By (\ref{eq:NB_bound1})-(\ref{eq:NB_bound2}), we have 
\begin{equation}\label{eq:NB_bound3}
Z(\mu^0  + \lambda v) = Z(\mu^0) + \lambda Z'(\mu^0,\nu) = L(\lambda) \leqslant UB(\lambda)  \text{ for all } \lambda \in \left[0,\lambda^\NB\right].
\end{equation}

\noindent This shows that the intersection value $\lambda^r \notin \left(0,\lambda^\NB\right)$. Depending on the value of $\lambda^r$, several cases need to be considered. If $\lambda^{r} \in [0,\bar{\lambda}]$, then $\lambda^r \geqslant \lambda^\NB$. By the concavity of function $Z$, we have 

\begin{equation}\label{eq:NB_bound4}
	c_{\mu^0+\lambda^r \nu}(C^*_{\mu^\NB}) \leqslant Z(\mu^0) + \lambda^r Z'(\mu^0,\nu) = L(\lambda^r) =  UB^r(\lambda^r).
\end{equation} 

\noindent Suppose that  $\lambda^{r} > \bar{\lambda}$. The case where $\lambda^{r} < 0$ can be handled similarly. By~(\ref{eq:NB_bound3}), $L(\lambda) < UB^r(\lambda)$ for any $\lambda \in [0,\bar{\lambda}]$. Again, by the concavity of function $Z$, we have

\begin{equation}\label{eq:NB_bound5}
c_{\mu^0+\bar{\lambda} \nu}(C^*_{\mu^\NB}) \leqslant  Z(\mu^0) + \bar{\lambda} Z'(\mu^0,\nu) = L(\bar{\lambda}) <  UB^r(\bar{\lambda}).
\end{equation}

The probability of randomly picking an edge in $\delta(C^*_{\mu^\NB})$ is 

\begin{align}\label{eq:NB_boundproba}
	Pr(e \in \delta(C^*_{\mu^\NB})) &= \frac{c_{\tilde{\mu}^r}(C^*_{\mu^\NB})}{c_{\tilde{\mu}^r}(E_r)} = 	\left\{
	\begin{array}{ll}
		\frac{c_{\mu^0+\lambda^r \nu}(C^*_{\mu^\NB})}{c_{\mu^0+\lambda^r \nu}(E_r)} \text{ if } \lambda^r \in [0,\bar{\lambda}]\\
		\frac{c_{\mu^0+\bar{\lambda} \nu}(C^*_{\mu^\NB})}{c_{\mu^0+\bar{\lambda} \nu}(E_r)} \text{ otherwise,} 
	\end{array}
	\right. \nonumber \\
	& \leqslant \frac{UB^r(\hat{\mu}^r)}{c_{\hat{\mu}^r}(E_r)} \text{ by (\ref{eq:NB_bound4}-\ref{eq:NB_bound5})} \nonumber\\
	&\leqslant \frac{2}{|V^r|} = \frac{2}{n-r+1}.
\end{align} 

\begin{lemma}\label{lem:NB_prob_hyper}
	Suppose that the next breakpoint exists. Any fixed optimal cut $C^*_{\mu^\NB}$ for edge costs $c_{\mu^\NB}(e)$ defining the slope $Z'(\mu^\NB,\nu)$ at $\mu^\NB$ is returned by Algorithm~\ref{Alg_nb_contract} with probability at least $\dbinom{n}{2}^{-1}$.
\end{lemma}
\begin{proof}
	Using (\ref{eq:NB_boundproba}),  the probability that cut $\delta(C^*_{\mu^\NB})$ survives all the edge contractions is at least
	
	$$(1-\frac{2}{n})(1-\frac{2}{n-1})\cdots(1-\frac{2}{3}) = \dbinom{n}{2}^{-1}.$$ \qed
\end{proof}

Note that this error probability is identical to the one for the original (non-parametric) contraction algorithm~\cite{Kar93,KS96}. 

If the next breakpoint does not exist, then there exists no cut $C$ such that the intersection value $\bar{\lambda}^\NB$ of functions $L(\lambda)$ and $c_{\mu^0 + \lambda \nu}(C)$ is nonnegative. In this case, Algorithm~\ref{Alg_nb_contract} gives a correct answer (with probability 1) after performing the test in Step~\ref{alg:NB_random_test}. Therefore, the success probability of Algorithm~\ref{Alg_nb_contract} is given as follows.

\begin{corollary}
	$\PNB(M)$ is solved by Algorithm~\ref{Alg_nb_contract} with probability at least $\dbinom{n}{2}^{-1}$.
\end{corollary}

The random edge selection of Algorithm~\ref{Alg_nb_contract} may be performed in $O(n)$ time by extending a technique given Karger and Stein~\cite{KS96}. With non-parametric costs, this technique is based on updating at each iteration $r$ of the algorithm a cost adjacency matrix $\varGamma =(\gamma(u,v))_{u,v \in V_r}$ and a degree cost vector $D =(d(v))_{v \in V_r}$ associated to graph $G_r$. The entries $\gamma(u,v)$ and $d(v)$ represent the cost of edge $(u,v)$ and the total cost of all the edges incident to node $v$, respectively. The update operations consists in replacing a row (a column) with the sum of two rows (columns) and removing a row and a column. Since the edges cost $c_{\tilde{\mu}^r}$ used to construct the probability distribution in Algorithm~\ref{Alg_nb_contract} may vary from one iteration to the other, the previous update operations are not possible and computing these matrices in $O(n^2)$ at each iteration is expensive. Instead, we may use two cost adjacency matrix and  degree cost vectors for costs $\bar{c}^0$ and $\bar{c}^1$. These matrices can be updated in $O(n)$ time as in \cite[Section 3]{KS96}. By embedding Algorithm~\ref{Alg_nb_contract} in the recursive scheme of Karger and Stein~\cite{KS96}, it follows that an optimal cut for our problem can be computed in $O(n^2\log^3n)$ time as for non-parametric costs. 

\begin{theorem}\label{thm:NB_rand}
	Problem $\PNB$ can be solved with high probability in $O(n^2\log^3n)$ time.
\end{theorem}

\section{Problem $\Pmax$}\label{sec:Pmax}

We give in this section an efficient algorithm for solving problem $\Pmax$ and show the following result. 

\begin{theorem}\label{thm:max}
	Problem $\Pmax$ can be solved in $O(\log^{d-1} (n)n^4)$ time.
\end{theorem}

Before detailing the algorithm, let us first introduce some notation. An {\em arrangement} $A(\mathcal{H})$, formed by a set $\mathcal{H}$ of hyperplanes in $\mathbb{R}^d$, corresponds to a partition of
$\mathbb{R}^d$ into $O(|\mathcal{H}|^d)$ convex regions called {\em cells}. See~\cite{Mul94} for more details. Given a polytope $P$ in $\mathbb{R}^d$, let $A(\mathcal{H})\cap P$ denote the restriction of the arrangement $A(\mathcal{H})$ into $P$. The following simplified version of standard problem in geometry called, \emph{point location in arrangements} (PLA), is used as a subroutine in our algorithm. This problem is solved by a multidimensional parametric search algorithm~\cite{CM93,Tok01}. See Appendix~\ref{sec:Pmax_geometric} for more details.

  \begin{enumerate}
	\item[$\Preg(\mathcal{H},P,\bar{\mu})$] Given a polytope $P$, a set $\mathcal{H}$ of hyperplanes in
	$\mathbb{R}^d$, and a target value $\bar{\mu}$, locate a simplex $R\subseteq A(\mathcal{H})\cap P$
	containing $\bar{\mu}$.
\end{enumerate}

\noindent Fix a constant $1 < \varepsilon < \sqrt{\frac{4}{3}}$ and let $\beta = \frac{\varepsilon^2-1}{ m} > 0$.
Compute $p = 1 + \lceil\log \frac{m^2}{\varepsilon^2-1} / \log \varepsilon^2 \rceil$
so that $\beta \varepsilon^{2(p - 1)} > m$, and observe that $p =
O(\log n)$. For a given edge $\bar{e} \in E$, define the $p+2$ affine functions $g_i :
\mathbb{R}^{d}\rightarrow \mathbb{R}$ by $g_0(\bar{e},\mu) = 0$,
\begin{equation*}
g_{i}(\bar{e},\mu)=\beta\,\varepsilon^{2(i-1)}c_{\mu}(\bar{e}) \quad\text{for } i=1,\dots,p,
\end{equation*}
\noindent and $g_{p+1}(\bar{e},\mu) = +\infty$.

\begin{algorithm}[h]
	\caption{Deterministic algorithm for $\Pmax$}\label{Alg_MAX}
	\begin{algorithmic}[1]
		\REQUIRE a graph $G=(V,E)$, costs $c^0,\ldots,c^d$, and $1 < \varepsilon < \sqrt{\frac{4}{3}}$ 
		\ENSURE the optimal value $\mu^*$
		\STATE let $A(\mathcal{H}_1)$ denote the arrangement formed by the set $\mathcal{H}_1$ of hyperplanes $H_{e,e'}=\{\mu \in \mathbb{R}^d: c_{\mu}(e)=c_{\mu}(e')\}$ for any pair of edges $e,e' \in E$   
		\STATE solve $\Preg(\mathcal{H}_1,M,\mu^*)$ and compute a simplex $R_1 \subseteq A(\mathcal{H}_1)\cap M$ containing $\mu^*$  \label{Alg:max_intersection1}
		\STATE  choose arbitrarily $\mu^1$ in the interior of $R_1$, compute a maximum spanning tree $T$ of $G$ for edge costs $c_{\mu^1}(e)$, and let $\bar{e}$ be an edge in $T$ such that $c_{\mu_1}(\bar{e}) = \arg \min_{e \in T} c_{\mu_1}(e)$ \label{Alg:max_ST}
		\STATE  let $\pi(e)$ denote the rank of edge $e \in E$ according to the increasing edges costs order in $R_1$ (ties are broken arbitrary) 
		\IF{$\min_{\mu \in R_1}c_{\mu}(\bar{e}) = 0$} \label{Alg:max_deterministic_if}
		\STATE let $\tilde{e}$ be an edge such that $\pi(\tilde{e}) \in \arg \min_{e \in E}\{\pi(e): c_e(\mu) > 0 \text{ for all } \mu \in R_1\}$ and $R'_1 = \{\mu \in R_1: g_p(\bar{e},\mu) \geqslant c_{\mu}(\tilde{e})\}$ \label{Alg:max_deterministic_if_positive_edge}
		\STATE set $R_1 \longleftarrow R'_1$ \label{Alg:NB_deterministic_if_remove}
		\ENDIF
		\STATE let $A(\mathcal{H}_2)$ denote the arrangement formed by the set $\mathcal{H}_2$ of hyperplanes $H_i(e)=\{\mu \in \mathbb{R}^d: c_{\mu}(e)=g_{i}(\bar{e},\mu)\}$ for any edges $e\in E$ and for $i=1,\ldots,p$   
		\STATE solve $\Preg(\mathcal{H}_2,R_1,\mu^*)$ and compute a simplex $R_2 \subseteq A(\mathcal{H}_2) \cap  R_1$ containing $\mu^*$  \label{Alg:max_intersection1}
		\STATE choose arbitrarily $\mu^2 \in R_2$ and compute the set $\mathcal{C}$ of all the $\varepsilon$-approximate cuts for edge costs $c_{\mu^2}(e)$ \label{Alg:NB_deterministic_end}
		\STATE let $A(\mathcal{H}_3)$ denote the arrangement formed by the set $\mathcal{H}_3$ of hyperplanes $H_{C,C'}=\{\mu \in \mathbb{R}^d: c_{\mu}(C)=c_{\mu}(C')\}$ for any pair of cuts $C,C' \in \mathcal{C}$   
		\STATE solve $\Preg(\mathcal{H}_3,R_2,\mu^*)$ and return $\mu^*$  \label{Alg:max_intersection1}
		\label{Alg:max_intersection1}	
	\end{algorithmic}
\end{algorithm}

One of the difficulties of parametric optimization is that edges are only partially ordered by costs. In order to overcome it, Algorithm~\ref{Alg_MAX} restricts the parametric search to a simplex $R_1$ containing $\mu^*$ where the parametric functions $c_\mu(e)$ are totally ordered. Let $\pi(e)$ denote the rank of edge $e \in E$ according to the increasing edges costs order in $R_1$. The algorithm needs to divide $R_1$ into smaller regions using as in Mulmuley~ \cite{Mul99} the relationship between cuts and spanning trees. However, the proof of Mulmuley's result is complicated and results in a large number of regions. Consider an arbitrary $\mu^1$ in the interior of $R_1$ and compute a maximum spanning tree $T$ for costs $c_{\mu_1}(e)$. Let $\bar{e}$ denote an edge in $T$ such that $c_{\mu_1}(\bar{e}) = \arg \min_{e \in T} c_{\mu_1}(e)$. Since  functions $c_\mu(e)$ may intersect only at the boundaries of $R_1$, for any edge $e \in T\setminus \{\bar{e}\}$ exactly one of the following cases occurs: $i)$  $c_{\mu^1}(e) = c_{\mu^1}(\bar{e})$, and therefore $c_{\mu}(e) = c_{\mu}(\bar{e})$ for all $\mu \in R_1$, or $ii)$ $c_{\mu^1}(e) > c_{\mu^1}(\bar{e})$, and therefore $c_{\mu}(e) \geqslant c_{\mu}(\bar{e})$ for all $\mu \in R_1$. In either cases, edge $\bar{e}$ satisfies

\begin{equation} \label{eq:ST}
c_{\mu}(\bar{e}) = \min_{e \in T} c_{\mu}(e) \text{ for all } \mu \in R_1.
\end{equation}

\noindent Since every cut in $G$ intersects $T$ in at least one edge, by (\ref{eq:ST}) we have the following lower bound on the minimum cut value. 

\begin{equation}\label{eq:max_lb}
c_{\mu}(\bar{e}) \leqslant Z(\mu) \text{ for all }  \mu \in R_1.
\end{equation}

\noindent Let $\bar{C}$ denote the cut formed by deleting $\bar{e}$ from $T$. By the cut optimality condition, we obtain the following upper bound on the minimum cut value. 

\begin{align}\label{eq:max_ub}
Z(\mu) &\leqslant c_\mu(\bar{C}) \nonumber \\
& = \sum_{e \in \delta(\bar{C})} c_\mu(e) \nonumber\\
&\leqslant m c_\mu(\bar{e})\nonumber\\
&< g_p(\bar{e},\mu),
\end{align}

\noindent where the last inequality follows from the definition of function $g_p(\bar{e},\mu)$.  Let $A(\mathcal{H}_2)$ denote the arrangement formed by the set $\mathcal{H}_2$ of hyperplanes $H_i(e)=\{\mu \in \mathbb{R}^d: c_{\mu}(e)=g_{i}(\bar{e},\mu)\}$ for any edges $e\in E$ and for $i=1,\ldots,p$. Suppose first that $c_{\mu}(\bar{e}) > 0$ for all $\mu \in R_1$, then by (\ref{eq:max_lb}) we have $Z(\mu) >0$ for all $\mu \in R_1$. In this case, we may apply the technique given in~\cite[Theorem 4]{AMTQMP} to compute all optimal cuts for edge costs $c_{\mu^*}(e)$. Consider a simplex $R_2 \subseteq A(\mathcal{H}_2) \cap R_1$ containing $\mu^*$ and any optimal cut $C^*_{\mu^*}$ for edge costs $c_{\mu^*}(e)$. Since functions $g_p(\bar{e},\mu)$ and $c_\mu(e)$, for all $e \in E$, may intersect only at the boundaries of $R_2$, it follows by (\ref{eq:max_ub}) that $c_\mu(e) \leqslant g_p(\bar{e},\mu)$ for all $e \in \delta(C^*_{\mu^*})$ and all $\mu \in R_2$. By construction of the arrangement $A(\mathcal{H}_2) \cap R_1$, for every edge $e$ in $\delta(C^*_{\mu^*})$ there exists some $q \in \{0,\ldots,p\}$ such that

\begin{equation}\label{eq:max_bounds}
	g_q(\bar{e},\mu) \leqslant c_\mu(e) \leqslant g_{q+1}(\bar{e},\mu) \mbox{ for all } \mu \in R_2.
\end{equation}

\noindent The following result shows that not all functions $c_\mu(e)$ of the edges in $\delta(C^*_{\mu})$ are below function  $g_{1}(\bar{e},\mu)$ for all $\mu \in R_2$.

\begin{lemma}\label{lem:max_lb}
	For any $\bar{\mu} \in R_2$ and any optimal cut $C^*_{\bar{\mu}}$ for edge costs $c_{\bar{\mu}}(e)$, there exists at least an edge $e \in \delta(C^*_{\bar{\mu}})$ satisfying $c_\mu(e) \geqslant g_{1}(\bar{e},\mu)$ for all $\mu \in R_2$.
\end{lemma}
\begin{proof}
	By contradiction, if the statement of the lemma does not hold, then 	
	$$Z(\bar{\mu}) \leqslant mg_1(\bar{e},\bar{\mu}) = (\varepsilon^2 -1) c_{\bar{\mu}}(\bar{e}) < c_{\bar{\mu}}(\bar{e}),$$
	
	\noindent where the last inequality follows from $\varepsilon^2 < 2$. This contradicts (\ref{eq:max_lb}) and thus, at least an edge $e \in \delta(C^*_{\bar{\mu}})$ satisfies $c_{\bar{\mu}}(e) \geqslant g_1(\bar{e},\bar{\mu})$. Since functions $g_1(\bar{e},\mu)$ and $c_\mu(e)$ may intersect only at the boundaries of $R_2$, we have $c_{\mu}(e) \geqslant g_1(\bar{e},\mu)$ for all $\mu \in R_2$. \qed
\end{proof}

\noindent By (\ref{eq:max_bounds}) and Lemma~\ref{lem:max_lb}, one can use the same arguments as in \cite[Theorem 4]{AMTQMP} and get the following result.

\begin{lemma}\label{lem:max_approx}
	If $Z(\mu) >0$ for all $\mu \in R_1$, then any specific optimal cut $C^*_{\mu^*}$ for edge costs $c_{\mu^*}(e)$ is an $\varepsilon$-approximate cut for edge costs $c_\mu(e)$ for every $\mu \in R_2$.
\end{lemma}

\noindent The optimal value $\mu^*$ is defined by the intersection of parametric functions $c_\mu(C)$ of at least $d$ optimal cuts $C$  for edge costs $c_{\mu^*}(e)$. If the condition of Lemma~\ref{lem:max_approx} holds, the enumeration of these solutions can be done by picking some $\mu^2$ in a simplex $R_2 \subseteq A(\mathcal{H}_2) \cap R_1$ containing $\mu^*$ and computing the set $\mathcal{C}$ of all the $\varepsilon$-approximate cuts for edge costs $c_{\mu^2}(e)$. Note that this set is formed by $O(n^2)$ cuts~\cite{NNI97}. Naturally, $\mu^*$ can be obtained by computing the lower envelope of the parametric functions $c_\mu(C)$ for all the cuts $C \in \mathcal{C}$. However, this  will take an excessive $O(n^{2d}\alpha(n))$ running time~\cite{EGS89}, where $\alpha(n)$ is the inverse of Ackermann's function. Instead, observe that $\mu^*$ is a vertex of at least $d$ cells of the arrangement $A(\mathcal{H}_3)\cap R_2$ formed by the set $\mathcal{H}_3$ of hyperplanes $H_{C,C'}=\{\mu \in \mathbb{R}^d: c_{\mu}(C)=c_{\mu}(C')\}$ for any pair of cuts $C,C' \in \mathcal{C}$. Therefore, $\mu^*$ is a vertex of any simplex containing it and included in a facet of $A(\mathcal{H}_3)$. By solving PLA problem in $A(\mathcal{H}_3) \cap R_2$, $\mu^*$ can be computed more efficiently. 

In order to complete the algorithm, we need to handle the case where $\min_{\mu \in R_1}c_{\mu}(\bar{e}) = 0$. It is sufficient to consider in this case a restriction $R'_1 \subset R_1$ containing $\mu^*$ such that $\min_{\mu \in R'_1}c_{\mu}(\bar{e}) > 0$. The following results show how to construct such a restriction. 

\begin{lemma}\label{lem:NB_positive_value}
	There exists at least an edge $\hat{e} \in E$ such that $c_{\mu}(e)  > 0$ for all  $\mu \in R_1$.
\end{lemma}
\begin{proof}
	Let $\hat{e}$ denote the edge maximizing the rank $\pi(e)$  among all the edges $e \in E$. Suppose that $\min_{\mu \in R_1}c_{\mu}(\hat{e}) = 0$ and let $\mu^0 \in R_1$ such that $c_{\mu^0}(\hat{e}) = 0$. The total order of edges costs in $R_1$ implies that $c_{\mu^0}(e)=0$ for all $e \in E$ and thus, we have $c_{\mu^0}(C)=0$ for any cut $C$ in $G$. Consider any $\mu^1 \neq  \mu^0$ in $R_1$. By the concavity of $Z$, we have 
	
	\begin{equation}\label{eq_max_concave}
	Z(\zeta \mu^0 + (1-\zeta) \mu^1) \geqslant \zeta Z(\mu^0) + (1-\zeta) Z(\mu^1)= (1-\zeta) Z(\mu^1) \text{ for all } \zeta \in [0,1].
	\end{equation}
	
	\noindent For all optimal cuts $C^*_{\mu^1}$ for edge costs $c_{\mu^1}(e)$, we have 
	
	$$c_{\zeta \mu^0 + (1-\zeta) \mu^1}(C^*_{\mu^1}) = \zeta c_{\mu^0}(C^*_{\mu^1}) + (1-\zeta) c_{\mu^1}(C^*_{\mu^1}) = (1-\zeta) Z(\mu^1).$$
	
	\noindent Therefore, by (\ref{eq_max_concave}) all optimal cuts $C^*_{\mu^1}$ are also optimal for edge costs $c_{\zeta \mu^0 + (1-\zeta) \mu^1}(e)$ for any $\zeta \in [0,1]$. Consider now the restriction of $Z$ to the segment $[\mu^0,\mu^1]$. By definition, if $\mu^1$ is a breakpoint of function $Z$, then there exists at least an optimal cut $C^*_{\mu^1}$ which is not optimal for edge costs $c_{\zeta \mu^0 + (1-\zeta) \mu^1}(e)$  for any $\zeta \in [0,1)$. Therefore, $\mu^1$ is not a breakpoint of $Z$. Since $\mu^1$ was chosen arbitrary, it follows that function $Z$ has no breakpoints in $R_1$. This leads to a contradiction since $\mu^*$ is a breakpoint of $Z$ in $R_1$. \qed
\end{proof}

\noindent Let $\tilde{e}$ be an edge such that $\pi(\tilde{e}) \in \arg \min_{e \in E}\{\pi(e): c_\mu(e) > 0 \text{ for all } \mu \in R_1\}$ and $R'_1 = \{\mu \in R_1: g_p(\bar{e},\mu) \geqslant c_{\mu}(\tilde{e})\}$. By Lemma~\ref{lem:NB_positive_value}, edge $\tilde{e}$ exists and we have $g_p(\bar{e},\mu) = \beta^2\,\varepsilon^{p-1}c_{\mu}(\bar{e}) > 0$ for all $\mu \in R'_1$. This shows that $c_\mu(\bar{e}) > 0$ for all $\mu \in R'_1$ and thus, the condition of Lemma~\ref{lem:max_approx} holds in $R'_1$. It remains now to show that $\mu^* \notin R_1 \setminus R'_1$.

\begin{lemma}\label{lem:NB_nobreakpoint}
	Function $Z$ has no breakpoint in $R_1 \setminus R'_1$. 
\end{lemma}
\begin{proof}
	Any edge $e \in E$ such that $\pi(e) \geqslant \pi(\tilde{e})$ satisfies $c_\mu(e) \geqslant g_p(\bar{e},\mu)$ for any $\mu \in R_1 \setminus R'_1$. Therefore, by (\ref{eq:max_ub}) no such edge is in an optimal cut for edge costs $c_{\mu}(e)$ for any $\mu \in R_1 \setminus R'_1$. Let $\check{e}$ be an edge such that $\check{e} = \arg \max_{e \in E} \{\pi(e): \pi(e) < \pi(\tilde{e})\}$. By the choice of $\tilde{e}$, there exist $\mu^0 \in R_1\setminus R'_1$ such that $c_{\mu^0}(\check{e}) =0$. The total order of edges costs in $R_1$ implies that $c_{\mu^0}(e) =0$ for all edges $e \in E$ such that $\pi(e) \leqslant \pi(\check{e})$. Therefore, we have $c_{\mu^0}(C)=0$ for any cut $C$ in $G$. Using the same argument as in the proof of Lemma~\ref{lem:NB_positive_value}, one can show that function $Z$ has no breakpoint in $R_1 \setminus R'_1$. \qed
\end{proof}

Let $T(d)$ denote the running time of Algorithm~\ref{Alg_MAX} for solving $\Pmax$ with $d$ parameters and $T(0) = O(n^2\log n + n m)$ denote the running time of computing a minimum (non-parametric) global cut using the algorithm given in \cite{SW97}. The input of a call to problem PLA requires $O(n^4)$ hyperplanes. Therefore, by Lemma~\ref{lem:geometric} given in the appendix, the $\Theta(1)$ calls to problem PLA can be solved recursively in $O(\log (n)T(d-1) + n^4)$ time. The enumeration of all the $O(n^2)$ approximate cuts can be done in $O(n^4)$ time \cite[Corollary 4.14]{NI08}. Therefore, the running time of Algorithm~\ref{Alg_MAX} is given by the following recursive formula.

\begin{align*}
	T(d) = O(\log (n)T(d-1) + n^4) &= O(\log^d (n) T(0) + \log^{d-1} (n)n^4) \\
										&= O(\log^{d-1} (n)n^4).
\end{align*}

\section{Conclusion}\label{conclusion}

As shown in \TAB{summary}, our improved algorithms are significantly
faster than what one could otherwise get from just using Megiddo.  As
mentioned in \SEC{results}, our results for $\PNB$ are close to being
the best possible, as they are only $\log$ factors slower than the
best known non-parametric algorithms.  One exception is that we don't
quite match Karger's \cite{Kar00} speedup to near-linear time for the
randomized case, and we leave this as an open problem.

Our deterministic algorithm for $\Pmax$ is also close to best
possible, though in a weaker sense.  It uses the ability to compute
all $\alpha$-optimal cuts in $O(n^4)$ time, and otherwise is only
$\log$ factors slower than $O(n^4)$.  The conspicuous open problem
here is to find a faster randomized algorithm for $\Pmax$ when $d >
1$.

We also developed evidence that $\PNB$ is in fact easier than $\Pmax$
for $d = 1$.  \SEC{oracle} showed that $\PNB$ is oracle-reducible to
$\Pmax$ for $d = 1$, and we were able to find a deterministic
algorithm for $\PNB$ that is much faster than our best algorithm for
$\Pmax$ when $d = 1$.

Finally, we note that Stoer and Wagner's \cite{SW97} algorithm was
generalized to {\em symmetric submodular function minimization} (SSFM)
by Queyranne \cite{Q98}.  Thus we could solve the equivalent versions
of $\PNB$ and $\Pmax$ for parametric SSFM by substituting Queyranne's
algorithm for Stoer and Wagner's algorithm in Megiddo's framework.
This leads to the question of whether one could find faster algorithms
for $\PNB$ and $\Pmax$ for parametric SSFM by generalizing our results.

\section{Appendix}

\subsection{Geometric tools}\label{sec:Pmax_geometric}

A classical problem in computational geometry called
\emph{point location in arrangements} (PLA) is useful to our algorithm. PLA has been widely used in various contexts such as linear programming~\cite{AST93,MS92} or parametric optimization~\cite{CM93,Tok01}. For more details, see~\cite[Chapter 5]{Mul94}.

Given a simplex $P$, an arrangement $A(\mathcal{H})$ formed by a set $\mathcal{H}$
of hyperplanes in $\mathbb{R}^d$, let $A(\mathcal{H})\cap P$ denote the restriction of the arrangement $A(\mathcal{H})$ to $P$. The goal of PLA is to construct a data structure in order to quickly locate a cell of $A(\mathcal{H})\cap P$ containing an unknown target value $\bar{\mu}$.  Solving PLA requires the explicit construction of the arrangement $A(\mathcal{H})$ which can be done in an excessive $O(|\mathcal{H}|^d)$ running time~\cite[Theorem 6.1.2]{Mul94}.  For our purposes, it is sufficient to solve the following simpler form of PLA. 

\begin{enumerate}
	\item[$\Preg(\mathcal{H},P,\bar{\mu})$] Given a simplex $P$, a set $\mathcal{H}$ of hyperplanes in
	$\mathbb{R}^d$, and a target value $\bar{\mu}$, locate a simplex $R\subseteq A(\mathcal{H})\cap P$
	containing a target and unknown value $\bar{\mu}$.
\end{enumerate}

\noindent  Cohen and Megiddo~\cite{CM93} consider the problem $Max(f)$ of maximizing a concave function $f: \mathbb{R}^d \rightarrow \mathbb{R}$ with fixed dimension $d$ and give, under some conditions, a polynomial time algorithm. This algorithm also uses problem $\Preg(\mathcal{H},P,\bar{\mu})$ as a subroutine, where in this context the target value $\bar{\mu}$ is the optimal value of $Max(f)$. Let $T(d)$ denote the time required to solve $Max(f)$ with $d$ parameters and $T(0)$ denote the running time of evaluating $f$ at any value in $\mathbb{R}^d$. The authors solve $\Preg(\mathcal{H},P,\bar{\mu})$ recursively using multidimensional parametric search technique. See also \cite{Clar87,CF90,Tok01}.

\begin{lemma}\label{lem:geometric}
	Given a simplex $P$, a set $\mathcal{H}$ of hyperplanes in $\mathbb{R}^d$, and a target and unknown value $\bar{\mu}$, $\Preg(\mathcal{H},P,\bar{\mu})$ can be solved in $O(\log (|\mathcal{H}|)T(d-1) + |\mathcal{H}|)$ time.
\end{lemma}

\subsection{Parametrized Stoer and Wagner's algorithm}\label{sec:Pmax_geometric}

In this appendix we discuss an application of the standard technique of Megiddo's parametric searching method~\cite{Meg79,Meg83} to global minimum cut in the context of Stoer and Wagner's algorithm (SW)~\cite{SW97}. Let us first recap the SW algorithm. For a fixed value $\bar{\mu}$ of the parameter $\mu$, the SW algorithm computes a \emph{Maximum Adjacency} (MA) ordering $(v_1,\ldots,v_{n})$ of the nodes.  The algorithm starts with an
arbitrary node $v_1$ and for each $i \in \{2,\ldots,n\}$ 
adds node $v_i \in V \setminus \{v_1,\ldots,v_{i-1}\}$ with maximum connection cost $f^{iv}(\bar{\mu}) := \sum_{j < i}\sum_{e \in E: e=(v_j,v)} c_{\bar{\mu}}(e)$ of all edges between $v \in
V \setminus \{v_1,\ldots,v_{i-1}\}$ and the set $\{v_1,\ldots,v_{i-1}\}$.  The key property is that cut $C^{(1)}=\{v_n\}$ is a
minimum cost cut separating $v_{n-1}$ and $v_n$. The SW algorithm stores this cut as a candidate for a
global min cut and merges nodes $v_{n-1}$ and $v_n$. The pair $(v_{n-1},v_n)$ is called a \textit{pendant pair}. If a global min cut separates $v_{n-1}$ and $v_n$ then
$C^{(1)}$ is an optimal global cut; otherwise $v_{n-1}$ and $v_n$ must be on the same side of any global minimum cut and thus this merging does
not destroy any global minimum cut. This process continues generating a sequence $C^{(1)},\ldots,C^{(n-2)}$ of candidate cuts (singletons in the contracted graphs) until $i = n-1$. At this step, the graph
contains only two vertices which give the final stored
candidate $C^{(n-1)}$.  The best cut in the set $\mathcal{S} = \{C^{(1)},\ldots,C^{(n-1)}\}$ is a global minimum cut.

\begin{definition}[\cite{CM93,Meg83,Meg84}]
	An algorithm $\mathcal{A}$ that computes function $Z$ for any $\mu \in M$ is called \textit{affine} if the operations that depend on $\mu$ used at each step are limited to additions, multiplications by constants, comparisons, and copies.
\end{definition}

Consider a target value $\bar{\mu} \in M$ which may correspond to the optimal value of problem $\PNB$ or $\Pmax$. We now show that the SW algorithm is affine.  Indeed, when $\mu$ is not fixed, the parametrized version of SW uses two types of comparisons. First, choosing the next node $v_i$ to add to an incomplete MA ordering $(v_1,\ldots,v_{i-1})$ requires computing the maximum of affine functions $f^{iv}(\bar{\mu})$ for all $v \in V \setminus \{v_1,\ldots,v_{i-1}\}$ for the target value $\bar{\mu}$. For this it suffices to compare affine functions $f^{iu}(\mu)$ and $f^{iv}(\mu)$ for all nodes $u$ and $v$ in $V \setminus \{v_1,\ldots,v_{i-1}\}$. The second type of comparison is to compute the minimum among of the costs $c_{\bar{\mu}}(C^{(i)})$ of candidate cuts for $i=1,\ldots,n-1$, which again amounts to comparisons between affine functions of $\mu$. 

In the following, we bound the running time of the parametrized MA
ordering. Recall that in contrast to problem $\Pmax$, problem $\PNB$ is a one-dimensional parametric optimization problem ($d=1$). Let $T(d)$ denote the time needed to solve the parametric optimization $\Pmax$ or $\PNB$ with $d$ parameters and $T(0) = O(n^2\log n + n m)$ denote the running time of computing a minimum (non-parametric) global cut using SW algorithm. For each iteration $i$ of SW algorithm, define $A(\mathcal{H}^i)$ the arrangement formed by the set $\mathcal{H}^i$ of hyperplanes $H^i_{u,v}=\{\mu \in \mathbb{R}^d: f^{iu}(\mu)=f^{iv}(\mu)\}$ for any pair of nodes $u,v \in V \setminus \{v_1,\ldots,v_{i-1}\}$ and define $R^0 = M$. Referring to the notation in Section~\ref{sec:Pmax} and using Lemma~\ref{lem:geometric}, problem $\Preg(\mathcal{H}^i,R^{i-1},\bar{\mu})$ can be solved in $O(\log (n)T(d-1) + n^2)$ time in order to compute a simplex $R^i \subseteq A(\mathcal{H}^i) \cap R^{i-1}$ containing the unknown target $\bar{\mu}$. By construction, functions $f^{iu}(\mu)$ are totally ordered in $R^i$. The next node $v_i \in V \setminus \{v_1,\ldots,v_{i-1}\}$ to be added to the incomplete MA ordering maximizes the connection cost $f^{iv}(\mu)$ for all $\mu \in R^i$ among the nodes $v \in V \setminus \{v_1,\ldots,v_{i-1}\}$. Therefore, node $v_i$ can be computed in $O(n)$ time. This shows that total running time of adding a node to an incomplete MA ordering is $O(\log (n)T(d-1) + n^2)$. After adding $O(n)$ vertices, the MA ordering can
be computed in $O(n\log (n) T(d-1) + n^3)$ time. The overall running time for computing the $O(n)$ MA orderings is $O(n^2\log (n)T(d-1) + n^4)$.

In the last step of the SW algorithm, we have a set $\mathcal{S} = \{C^{(1)},\ldots,C^{(n-1)}\}$ of candidate cuts. We now determine the target value $\bar{\mu}$ as follows. By construction, at least one cut $C^{(i)} \in \mathcal{S}$ corresponds to the target value $\bar{\mu}$, i.e., is minimum for edge costs $c_{\bar{\mu}}(e)$. By the correctness of the
algorithm, at least a solution in $\mathcal{S}$ is optimal for edge costs
$c_\mu(e)$ for any $\mu \in R^{n-1}$. This yields a simple approach for
solving $\PNB$. Compute the minimum intersection
point $\bar{\lambda}$ of function $Z(\mu^0)+\lambda Z'(\mu^0,\nu)$ and the $O(n)$ 
functions $c_{\mu^0+ \lambda v}(C^{(i)})$ for all $C^{(i)} \in
\mathcal{S}$ such that $\bar{\lambda} >0$. If such a point exists,
then the next breakpoint $\mu^\NB  =  \mu^0+ \bar{\lambda} v$. Otherwise, the next breakpoint does not exist.

For problem $\Pmax$, observe that $\mu^*$ is a vertex of the arrangement $A(\mathcal{H}^*)$ formed by the set $\mathcal{H}^*$ of hyperplanes $H^{ij}=\{\mu \in \mathbb{R}^d: c_\mu(C^{(i)})=c_\mu(C^{(j)})\}$ for any pair of cuts $C^{(i)},C^{(j)} \in \mathcal{S}$. Therefore, $\mu^*$ is a vertex of any simplex containing it and included in a facet of $A(\mathcal{H}^*)$. By Lemma~\ref{lem:geometric}, a simplex $R^* \subseteq A(\mathcal{H}^*) \cap R^{n-1}$ with a vertex $\mu^*$ can be computed in $O(\log (n)T(d-1) + n^2)$ time. 

The following proposition summarizes the running time for solving problems $\Pmax$ and $\PNB$ by this approach.  

\begin{proposition}\label{prop:meggido}
	Megiddo's parametric searching method combined with the SW algorithm solves problem $\PNB$ in $O(n^{5}\log (n))$ time and $\Pmax$ in $O(n^{2d+3}\log^d (n))$ time. 
\end{proposition}
\begin{proof}
	For problem $\PNB$, the overall running time is $T(1) = O(n^2\log (n) T(0) + n^4) = O(n^{5}\log (n))$ as claimed since $m = O(n^2)$. For problem $\Pmax$, the running time is given by the following recursive formula:
	
	\begin{align*}
	T(d) &= O(n^2\log (n) T(d-1) + n^4)\\
	&= O((n^2\log (n))^{d} T(0) + n^4 \sum_{i=0}^{d-1} (n^2\log n)^{i})\\
	&= O(n^{2d+3}\log^d (n)).
	\end{align*}
\end{proof}

\end{document}